\documentclass[]{article}
\usepackage{amsmath, amssymb, amsthm, graphicx, hyperref}
\newtheorem{theorem}{Theorem}
\newtheorem{lemma}[theorem]{Lemma}

\title{Decentralization Cheapens Corruptive Majority Attacks}

\author{Stephen H. Newman\footnote{Princeton University. Email: \href{mailto:stephen.newman@princeton.edu}{stephen.newman@princeton.edu}}}

\newcommand{\R}{\mathbb{R}}
\newcommand{\E}{\mathbb{E}}

\newcommand{\defeq}{\overset{\text{\tiny def}}{=}}

\newcommand{\pool}{M}  
\newcommand{\pf}{\phi}  
\newcommand{\tpow}{\Phi}  
\newcommand{\reward}{R}  
\newcommand{\atime}{T}  
\newcommand{\valmul}{v}  
\newcommand{\probflip}{f}  
\newcommand{\probthresh}{f_{\max}}
\newcommand{\powthresh}{\gamma}

\newcommand{\gdpowfrac}{g_{\mathrm{Def}}}
\newcommand{\specpayout}{c}

\DeclareMathOperator*{\argmin}{\arg\min}

\begin{document}
	
\maketitle

\begin{abstract}
	Corruptive majority attacks, in which mining power is distributed among miners and an attacker attempts to bribe a majority of miners into participation in a majority attack, pose a threat to blockchains. Budish bounded the cost of bribing miners to participate in an attack by their expected loss as a result of attack success. We show that this bound is loose. In particular, an attack may be structured so that under equilibrium play by most miners, a miner's choice to participate only slightly affects the attack success chance. Combined with the fact that most of the cost of attack success is externalized by any given small miner, this implies that if most mining power is controlled by small miners, bribing miners to participate in such an attack is much cheaper than the Budish bound. We provide a scheme for a cheap corruptive majority attack and discuss practical concerns and consequences.
\end{abstract}


\textbf{Acknowledgments:} Thanks to Matt Weinberg for substantial discussion, feedback, and advice.

\pagebreak

\section{Introduction}

Blockchain- and consensus-based ledger protocols are generally susceptible to \textit{majority attacks}, in which an attacker gains control of a majority of mining power and uses it to create a new canonical transaction history which diverges substantially from the original heaviest chain \cite{nakamoto2008bitcoin}. This attack may be highly profitable: on currency-only blockchains, this enables doublespend attacks and may cause chaos and/or devaluation, all of which may be used for economic gain. On blockchains that implement higher-level protocols and applications, such as Ethereum, attackers may also retroactively alter the state of smart contracts or other time-varying constructs, with similar consequences.

Historically, concerns about majority attacks have been dismissed as irrelevant to the current state of major cryptocurrencies. Budish argues that this is not necessarily the case in the long run (or even currently): the cost of an attack on the blockchain is proportional to the mining payout rate, so large transaction flow relative to this cost makes majority attacks profitable \cite{budish2018economic}. We show that under simple assumptions about the distribution of miner powers, the situation is far \textit{worse} than Budish's upper bound suggests. Budish bounds the necessary payout to miners as the cost to them of attack success. While the cost of attack success is indeed well-estimated by Budish, individual miners' actions are typically not substantially causally correlated with attack success, and so it suffices to pay miners their cost of attack success \textit{times the marginal increase in attack success probability that resulted from their actions}. As a result, in the by-design scenario where miners are small and therefore no individual or small group can exercise substantial control over a blockchain (though this is not always the stable state of affairs \cite{arnosti2022bitcoin}), corruptive majority attacks are both cheap and hard to prevent for proof-of-work blockchains, implying that cryptocurrencies are less game-theoretically stable than previously believed.

After a brief exposition of our model, we develop a simple framework for miner incentive analysis that illuminates the action/result-correlation aspect of their incentive problem. In particular, we show that for small miners, costs of participation in an attack are mostly externalized and sometimes very small. We then develop a more rigorous model of corruptive attacks, in which miners can change their behavior from timestep to timestep depending on state of the attack, and show that appropriately structured attacks succeed and are cheap with high probability. We give a practical example of such an attack on Bitcoin, including calculation of expected cost and discussion of profitability. We discuss novel and previously proposed economic-incentive-based prevention and mitigation strategies. We also give a brief overview of some of the non-economic incentives that affect the feasibility/likelihood of such an attack.

\subsection*{Related Work}

The consideration of majority attacks against digital currencies dates back to at least Nakamoto's work \cite{nakamoto2008bitcoin}. Bonneau noted the danger of bribery-based attacks and discussed attack methodologies and potential countermeasures in \cite{bonneau2016buy}. Since then, a variety of game-theoretic attack techniques (e.g. \cite{liao2017incentivizing}) and practical methodologies (e.g. \cite{mccorry2019smart}) have been proposed. Judmayer et al. provide a broad overview of bribery attack modeling in \cite{judmayer2021sok}, including majority attacks. Several majority attacks have also been conducted against a variety of cryptocurrencies, with varying results, as summarized (as of 2019) by \cite{shanaev2019cryptocurrency}.

The central -- and often ignored -- obstacle in conducting a majority attack is the long-term cost to miners resulting from currency devaluation stemming from the attack, as first noted in \cite{nakamoto2008bitcoin}. Budish analyzed this concept in detail, providing a formal model of the cost to a majority miner and approximating it for Bitcoin in \cite{budish2018economic}. Moroz et al. responded with a model of an attack-counterattack game with a liquid mining power marketplace, arguing that in the Budish setting, the threat of counterattack served to deter majority attacks designed to allow doublespending \cite{moroz2020double}.

There is also growing interest in the effects of incentive manipulation on a wider class of social-choice mechanisms without money. For instance, bribery \cite{faliszewski2016control} and coalition effects \cite{xia2023impact} have been studied in the context of voting, and there is continuing investigation into similar effects in matching markets and other contexts.

\section{Modeling Mining and the Cost of Corruption}

\subsection{Mining Model}
We assume a model of mining and miner incentives similar to that of Budish \cite{budish2018economic}. We assume a fixed set $\pool$ of miners, with a function $\pf:\pool\to \R_{\geq 0}$ mapping miners to their powers (hashrates in the case of hash-based PoW coins, for instance), and denote the total power $\tpow\defeq \sum_{m\in \pool}\pf(m)$.
We assume that each block mined is mined by a random miner, chosen at time of mining, with independent probability $\frac{\pf(m)}{\tpow}$ of selecting miner $m$ to mine any given block. A miner who mines a block receives reward $R$ for doing so. We assume that miners may choose to dedicate mining power to an attacker whose goal is to execute a majority attack. For simplicity (and optimistically for stability against such attacks), we assume that miners who attack on a mining turn receive no block rewards that turn.

\subsection{Miner Valuation Model}
We assume that miners extract value exclusively from present and future mining activity. In particular, we will consider two primary costs to miners of selling mining power: the direct expected lost revenue from lost time mining, and the expected lost future revenue from the increased chance of a majority attack as a result of providing hashpower. We will assume (pessimistically for the attacker) that a successful majority attack will cost a miner exactly their time-discounted expected future revenue\footnote{This corresponds to the assumption that a successful attack will prevent any future mining, while a failed attack does not change the value of future mining, causing the largest possible losses to miners.}. This will be indicated by $\valmul R$ for some value $\valmul$ to be bounded later.

As in Budish's work, the primary cost of attacks will lie in compensating miners for the change in the likelihood of attack success due to their participation -- the cost of corruption. We wish to bound this.

\subsection{Budish's Cost of Corruption}

Budish argued (and we agree) that the cost to miners caused by an attack is bounded by a sum of two costs. The more obvious component is their expected \textit{immediate loss}: their lost mining revenue as a result of devoting mining time to the attack. This may be trivially upper-bounded as \[\atime\frac{\pf(m)}{\tpow} \reward\] for an attack stretching over $\atime$ timesteps, where $\reward$ is the block reward, and will be negligible compared to the second component of the loss in most cases.

The more subtle loss that a miner incurs is their expected \textit{devaluation loss}: their time-discounted future loss of revenue as a result of currency collapse. Following Budish's analysis, we may bound this by \[\valmul \frac{\pf(m)}{\tpow}\reward\] for some constant multiplier $\valmul$.\footnote{We assume that a coin will have no change in value under attack failure and will undergo total and permanent devaluation under attack success. This represents an optimistic assumption from the point of view of stability, as it drives the expected cost to miners of an attack (and therefore the cost of said attack) as high as possible.} For instance, assuming constant distribution of mining power for a cryptocurrency that pays $k$ fixed block rewards $\reward$ per year, $\valmul = \frac{1}{1-\sqrt[k]{0.95}}\approxeq 20k$ would reflect the sum of expected income for all time with 5\%/year discounting applied. In most cases, however, $\valmul$ is substantially lower. For instance, assuming that $\reward$ halves every few years (or that the total mining power of other miners is large and doubles every few years) we may bound $\valmul$ as the number of blocks in just a few years -- and assumptions about ongoing costs of mining and revenue from selling mining equipment in a crash further reduce this estimate.

Combining the above two losses, we may weakly incentivize a miner $m\in \pool$ to participate in an attack of duration $\atime$ by paying them $(\atime+\valmul) \frac{\pf(m)}{\tpow} \reward$. Summing across miners, we see that it costs \[(\atime + \valmul)\reward\] to pay all miners to participate.

As Budish points out, this may already be dangerous. $\atime \reward$ is low compared to potential profit, and $\valmul$ may be as well -- for blockchains where $\tpow$ is expected to rise rapidly (PoW coins, for instance), gross mining income for given hardware is expected to drop proportionately, implying that almost all income is obtained in the next few years of mining. For instance, assuming $\valmul\approxeq 2\cdot 10^5$, the number of Bitcoin blocks in a little under four years, we get an attack price above $1.2\cdot 10^6$ BTC, slightly over 5\% of BTC in existence. Even in cases where mining power does not substantially increase (some PoS protocols, for instance) and currency valuation is expected to remain approximately constant, $\valmul$ may be bounded by the number of blocks in twenty years, as $\valmul \frac{\pf(m)}{\tpow}R$ is then enough to purchase a stock market portfolio that, if it yields 5\%/year over inflation, will be more profitable than mining at relative power $\frac{\pf(m)}{\tpow}$ assuming constant currency valuation.

\subsection{The Participation-Success Matrix}
We first consider a binary model for action during an attack: during an attack, a player $p$ may either \textit{participate} or \textit{refuse} to participate, and the attack will succeed or fail with probabilities determined by the participation/refusal of the various miners. Player $p$'s reward will depend on both participation/refusal and success/failure.

We note something simple but interesting: when a player's reward depends mostly on success or failure, and success/failure is only slightly affected by participation/refusal, the marginal cost of participation/refusal is much lower than the difference in reward between the success and failure cases. Consider, for instance, a small miner who participates in a majority attack that is expected to permanently crash Bitcoin if it succeeds. Assume that their expected time-discounted future valuation was $U$, and that their expected mining reward over the course of the attack period is $u\ll U$. Then, excluding order-$u$ increases in payout due to difficulty reduction, their mining payoff matrix is 
\begin{center}
	\begin{tabular}{c||c|c|}
		& Attack succeeds & Attack fails\\
		\hline \hline
		Participate & $0$ & $U$\\
		\hline
		Refuse & $u$ & $U+u$\\
		\hline
	\end{tabular}
\end{center}

For instance, considering the cases where the miner is sure that the attack will (succeed/fail) regardless of their action, their expected loss as a result of participating in the attack is bounded by ($0$/$u$). Relaxing this slightly, if they are sure that their participation affects attack success chance by at most $\epsilon$, their expected losses are at most $u+\epsilon U$.

\subsection{Bounding Expected Devaluation Loss}
The difference between our bound on cost of corruption and Budish's is simple: while Budish proposed paying all miners their expected losses from an attack, we propose paying them only the \textit{marginal increase} in their expected losses as a result of their participation in the attack. We consider an explicitly multiparty attack: the attacker attempts to corrupt many smaller miners into attack participation, and therefore must pay them only the damage they expect to cost themselves by joining. Under this view, Budish's cost of corruption is not tight -- it pays miners enough to convince them to join even in a case where act/no-act and success/failure are perfectly correlated. This is generally far from the truth: the smaller a miner is, the less impact their choice has on their perceived likelihood of attack success (and therefore their expected loss).

Let $\probflip(\frac{\pf(m)}{\tpow}):[0, 1]\to [0, 1]$ be a bound on the probability, as estimated by a miner $m$, that miner $m$'s participation in the attack will cause it to succeed (i.e. that it will not succeed if they do not participate, and will succeed if they do). Then we may tighten the bound on their devaluation loss to $\probflip(\frac{\pf(m)}{\tpow}) \frac{\pf(m)}{\tpow} \reward \valmul$ and the bound on their total loss to \[\left(\atime + \probflip\left(\frac{\pf(m)}{\tpow}\right)\valmul\right)\frac{\pf(m)}{\tpow}\reward\]

\section{A Thresholding Corruptive Attack}
We now analyze the attack implied by the above payout rule. Assume for the moment that we can verify full participation in an attack. Pick a start time for the attack, and pay any miner $m$ who participates in the attack \[\left(\atime + \probthresh\valmul\right)\frac{\pf(m)}{\tpow}\reward\] for $\probthresh$ chosen such that $\sum_{m:\probflip\left(\frac{\pf(m)}{\tpow}\right)\leq \probthresh} \pf(m)$ is substantially greater than $\frac{\tpow}{2}$. If each miner acts rationally, any miner $m$ with $\probflip(\frac{\pf(m)}{\tpow})\leq \probthresh$ will participate, and our attack will succeed. Summing across all miners, total cost of the attack is bounded by \[\left(\atime + \probthresh \valmul\right) \reward\]
which for small values of $\probthresh$ is much smaller than the attack cost under Budish's analysis.

\subsection{Estimating $\valmul$}

$\valmul$ depends on changing economic conditions, the specifics of the protocol under attack, and other real-world factors \cite{budish2018economic}. In Bitcoin and other PoW coins, for instance, $\valmul$ equal to the number of blocks in a few years (i.e. $\frac{\pf(m)}{\tpow} \valmul R$ equal to undiscounted gross earnings over the next few years) may be a reasonable estimate, given reward halving, electricity costs, continual improvements in mining hardware, and regulatory concerns. Notably, payouts from attack participation provide substantial security as compared to mining, for the same real-world reasons as cited above. For more in-depth discussion of $\valmul$, see \cite{budish2018economic}.

\subsection{Estimating $\probflip$}
We have two means of bounding $\probflip$. First, under the assumption that \[\left(\sum_{m\in \pool:\frac{\pf(m)}{\tpow}\leq \probthresh} \frac{\pf(m)}{\tpow}\right)\gg\frac{1}{2},\] $\probflip\left(\frac{\pf(m)}{\tpow}\right)$ is likely bounded on the close order of $\frac{\pf(m)}{\tpow}$, as the expected distribution of participants should not concentrate except possibly at nearly all players participating or nearly no players participating (in which case $\probflip(m)$ is very low). As such, $\probthresh$ is only required to be big enough to make the above fraction substantially greater than $\frac{1}{2}$, and will therefore be small for a well-decentralized blockchain. The cheaper attack price that results from smaller miners under this bound corresponds to the fact that these smaller miners externalize more of the damage caused by participating -- a miner's expected loss from participation is quadratic in their power, but their expected damage is linear. 

Second, and more concerningly, if we can convince all miners that the attack is destined to succeed, $\probflip$ becomes 0, as no individual miner's choice will affect the attack. Convincing miners that the attack will almost surely succeed likewise forces $\probflip$ very low.

Both of these dynamics will prove relevant in the next model.

\section{Refining the Model: A Block-By-Block Attack}
\label{sec:complicated-attack}

The previous model relied on the assumption that miner behavior over the course of an attack will not change, and assumed a somewhat arbitrary bound on $\probflip$. We now consider a multiparty refinement of the classical model for majority attacks. An attack is modeled as a $\atime$-step game with chance. The state of the game at the conclusion of timestep $t\in \{1, \dots, \atime\}$ is given by an integer $l_t$ indicating the length discrepancy between the attacking and defending chains (positive if the defending chain is longer), with $l_0$ the distance from the attack fork point to the tip of the heaviest chain at the start of the attack.

The attacker wins the game at the first timestep $t$ where $l_t=0$; if no such timestep has occurred by $t=\atime$, the attacker loses. As in \cite{judmayer2021sok} (in the case of zero native attacker power), miners are partitioned into two groups. We assume that a substantial set of the miners are honest (will always defend), either because we have not provided them sufficient economic incentive to attack, as will be the case for very large miners, or because of non-economic factors. We will model the remaining non-committed miners as rational agents that choose to mine for either the attacker or the defender based on expected time-discounted future rewards. At each step $t$, all miners decide to either attack or defend, and a random miner is chosen with probability proportional to their mining power. If that miner is defending, they mine a block on the canonical chain, setting $l_{t}=l_{t-1}+1$. If attacking, they mine a block on the attacking chain, setting $l_t = l_{t-1}-1$, and they receive a payout $\specpayout_{t-1, l}$ from the attacker for their choice to attack. Equivalently, we may model miners as making the choice to attack/defend after they are selected as the current round's block miner. At the end of the game, if the attacker lost, each miner $m$ receives payout $\valmul\frac{\pf(m)}{\tpow} \reward$, their expected future mining returns.

Our payout rule will not incentivize participation by miners of size $>\powthresh\tpow$, so we assume for simplicity that all such miners are part of the honest pool -- though if they do, it will aid the attack. We assume that the honest miners have combined mining power $\leq \gdpowfrac\tpow$. We will again exploit the fact that the actions of small miners (of power $\leq \powthresh\tpow$, for our purposes) are only slightly correlated with attack success chance to construct a cheap attack.

\subsection{A Payout Rule With Unique Subgame Perfect Nash Equilibrium}

As before, our payouts will come in two varieties. Every time a miner mines a block on the attacking chain (as opposed to the defending chain), they lose their block reward $\reward$ that would have been paid out on the defending chain. We therefore provide them a payout of $\reward$ \textit{via the defending chain} every time they mine an attacking block. This means that, no matter the eventual state of the attack, their block-reward payouts are identical across all of their strategy options, and their net rewards therefore follow those of the attacking game described above up to a constant. To incentivize miners to participate in the attack, it therefore suffices to provide a payout rule that enforces a unique Nash equilibrium of participation in said game. From an intuitive point of view, we hope to create a payment rule that causes the attack to have a very high chance of succeeding regardless of the behavior of any individual miner. This, in turn, should allow us to pay each individual miner very little, as their marginal loss as a result of participation will be low.

We first define a pseudo-value function 
\begin{align*}
	w_{t, l}^{\max} = \begin{cases}
		l = 0: & 0\\
		t=\atime, l > 0: & \valmul \powthresh \reward\\
		t<\atime, l>0: & \gdpowfrac w_{t+1, l + 1}^{\max} + (1-\gdpowfrac) w_{t+1, l - 1}^{\max} + \powthresh (w_{t+1, l + 1}^{\max} - w_{t+1, l - 1}^{\max})
	\end{cases}
\end{align*}

This is motivated by the following theorem:

\begin{theorem}
	\label{theorem:ne}
	Fix a payout scheme where the miner who mines block $t$ from the initial state $(t-1, l)$ receives (if they mined the block on the attacking chain) $\specpayout_{t-1, l} = (w^{\max}_{t, l + 1} - w^{\max}_{t, l - 1})$. This scheme admits a unique subgame perfect Nash equilibrium in the mining game described earlier, under which all non-committed miners participate in the attack on every block.
\end{theorem}
\begin{proof}
	
	We define the value function for miner $m$ as
	\begin{align*}
		w_{t, l}(m) = \begin{cases}
			l = 0: & 0\\
			t=\atime, l > 0: & \valmul \frac{\pf(m)}{\tpow} \reward\\
			t<\atime, l>0: & \gdpowfrac w_{t+1, l + 1} + (1-\gdpowfrac) w_{t+1, l - 1} + \frac{\pf(m)}{\tpow} c_{t, l}
		\end{cases}
	\end{align*}
	
	We induct backwards from $t=\atime$ on the joint hypotheses that:
	\begin{enumerate}
		\item $w_{t, l}(m)$ is equal to the expectation of the sum of attack payouts and time-discounted post-attack mining rewards for a miner $m$ in the non-committed group playing the subgame perfect Nash equilibrium strategy.
		\item Any non-committed miner selected at time $t$ is incentivized to participate in the attack at step $t$ given the above payout.
	\end{enumerate}
	
	The first claim is trivial at $t=\atime$, as the attack's success/failure is determined at the end of step $\atime$, and our choice for $w_{\atime, l}$ reflects the expected future rewards under those circumstances. The second claim at $t=\atime$ follows fairly simply: if their decision would determine attack success, they are paid $\valmul\powthresh\reward> \valmul \frac{\pf(m)}{\tpow} \reward$ (their expected loss from attack success), and if not, they are paid nothing (for uniqueness, assume a very small payment in this case).
	
	Given that the hypotheses hold for $t>t_0$, we may show that hypothesis 1 holds for $t=t_0$: observe that as every non-committed miner will attack on the next step, the state $(t_0, l)$ has chance $\gdpowfrac$ of evolving to $(t_0+1, l+1)$ and chance $(1-\gdpowfrac)$ of evolving to $(t_0+1, l-1)$. Moreover, the chance that miner $m$ will be selected to mine the next block (and therefore reap an additional reward of $c_{t_0-1, l}=(w^{\max}_{t_0, l + 1} - w^{\max}_{t_0, l - 1})$) is $\frac{\pf(m)}{\tpow}$.
	
	It remains to show hypothesis 2 for $t=t_0$ given hypothesis 1 for $t\geq t_0$ and hypothesis $2$ for $t>t_0$. We may observe that by hypothesis 1, the loss incurred for participating on step $t_0$ is $(w_{t_0, l+1}(m) - w_{t_0, l-1}(m))$. It therefore suffices to demonstrate that this is less than $w^{\max}_{t_0, l+1} - w^{\max}_{t_0, l-1}$. $w_{t, l}(m)$ is in fact positive-linear in $\frac{\pf(m)}{\tpow}$, and so as it is increasing in $l$, $(w_{t_0, l+1}(m) - w_{t_0, l-1}(m))$ is positive-linear in the same. At $\frac{\pf(m)}{\tpow}=\powthresh$, we have $w_{t, l}(m)=w^{\max}_{t, l}$ and so $(w_{t_0, l+1}(m) - w_{t_0, l-1}(m))=(w^{\max}_{t_0, l+1} - w^{\max}_{t_0, l-1})$; therefore, for $\frac{\pf(m)}{\tpow}\leq \powthresh$, we have $(w_{t_0, l+1}(m) - w_{t_0, l-1}(m))\leq (w^{\max}_{t_0, l+1} - w^{\max}_{t_0, l-1})$ as desired.
\end{proof}

The attacker must be able to credibly commit to their payout scheme to make the reward scheme accurate (and therefore for the equilibrium to hold). This is reasonable -- presuming that their payout for attack failure is at least $\valmul\powthresh\reward$ plus the sum of their payouts (which we will see is small), they are incentivized to keep paying under the Nash equilibrium regardless of the game state (as their reward is inverse to that of a miner with power $\powthresh\valmul$). This may also be generalized to arbitrary behavior, assuming that the attacker and the participant miners can agree on an expectation of future miner behavior.

\subsection{Success Likelihood, Expected Attack Length, and Expected Attack Cost}
We assume $\gdpowfrac+\powthresh<\frac{1}{2}$. Attack success probability is at least the probability that, had the attack been allowed to continue for $\atime$ steps whether or not $l_t$ became 0, $l_\atime$ would have been $\leq 0$. We know that $l_t$ decreases with probability $1-\gdpowfrac$ and increases with probability $\gdpowfrac$, so after $\atime$ steps, it has expectation $l_0- \atime(1-2\gdpowfrac)$ (assuming counterfactually that we continued the attack after $l_t=0$). Then for $\atime> \frac{l_0}{1-2\gdpowfrac}$, Hoeffding's inequality gives that the probability that $l_\atime>0$ is \[\leq \exp\left[- \frac{\frac{1}{2}\left(l_0- \atime(1-2\gdpowfrac)\right)^2}{\atime}\right]\]

Expected time to attack conclusion is bounded by the expected stopping time of the biased random walk described above, equal to $\frac{l_0}{1-2\gdpowfrac}$.

We now prove a generic bound on $\specpayout_{t-1, l}$. First, let $\left(X^{t, l}\right)_{t, l}$ be a $(t, l)$-indexed collection of random walks on the integers s.t. $X^{t, l}$ starts at time $t$ at position $l$, goes to time $\atime$, and increases/decreases at each step with probabilities $(\gdpowfrac+\powthresh)$, $(1-\gdpowfrac-\powthresh)$ respectively. For $i\geq t$, let $X^{t, l}_i$ be the position of $X^{t, l}$ at time $i$. The choice of the step probabilities gives us that, by its definition, $w^{\max}_{t, l} = \valmul\powthresh\reward \Pr\left[\min_{i\in \{t, \dots, \atime\}} X^{t, l}_i \leq 0\right]$. Now by coupling the walks $X^{t, l+1}$ and $X^{t, l-1}$ so that one increases on step $i$ iff the other does, we have
\begin{align*}
	w^{\max}_{t, l+1} - w^{\max}_{t, l-1} &= \valmul\powthresh\reward\Pr\left[\min_{i\in \{t, \dots, \atime\}} X^{t, l+1}_i \in \{1, 2\}\right]\\
	&=\valmul\powthresh\reward\sum_{\tau=t}^\atime \Pr\left[X^{t, l+1}_\tau\in \{1, 2\}\right] \Pr\left[\tau=\argmin_{i\in \{t, \atime\}} X^{t, l+1}_i|X^{t, l+1}_\tau\in \{1, 2\}\right]
\end{align*}

Each of the terms in the sum may be bounded. Proofs of these results are contained in the appendix.
\begin{lemma}
	\label{lem:probexactbound}
	\[\Pr\left[X^{t, l+1}_\tau\in \{1, 2\}\right] \leq \frac{1}{\sqrt{\tau - t}} \frac{\left(1-\gdpowfrac - \powthresh\right)^{5/2}}{\sqrt{2\pi}\left(\gdpowfrac + \powthresh\right)^{7/2}}\]
\end{lemma}

\begin{lemma}
	\label{lem:probminbound}
	\[\Pr\left[\tau=\argmin_{i\in \{t, \atime\}} X^{t, l+1}_i|X^{t, l+1}_\tau\in \{1, 2\}\right]\leq e^{-\frac{1}{2}(\atime-\tau)(1-2\gdpowfrac-2\powthresh)^2}\]
\end{lemma}

We also require a technical lemma:
\begin{lemma}
	\label{lem:tech1}
	Let $0<a\leq 1$. Then \[\sum_{i=t}^T \frac{e^{-a(T-i)}}{\max(\sqrt{i-t}, 1)}\leq \min\left(\frac{2}{\sqrt{T+1-t - 2\frac{\ln (1+\sqrt{T+1-t})}{a}}}, 1+ \frac{1}{1-e^{-a}}\right)\]
\end{lemma}

Combined, these yield a straightforward bound on worst-case total attack cost:
\begin{theorem}
	\label{thm:maxcost}
	The worst-case cost of the $w^{\max}$-attack is bounded by 
	\begin{align*}
		\atime \reward + \sum_{t=1}^\atime c_{t-1, l_t} &\leq \atime \reward + \valmul \powthresh\reward\left[\frac{4 \ln (1+\sqrt{\atime})}{(1-2\gdpowfrac - 2\powthresh)} + \sqrt{\frac{2}{\pi}} \frac{(1-\gdpowfrac - \powthresh)^{5/2}}{(\gdpowfrac + \powthresh)^{7/2}} 2 \sqrt{\atime}\right]
	\end{align*}
\end{theorem}

We also have a much lower bound on expected attack cost. 

\begin{theorem}
	\label{thm:expectedcost}
	The expected cost of the $w^{\max}$-attack is bounded by \[\reward\left(\frac{l_0}{2} + \frac{1}{2}\frac{l_0}{1-2\gdpowfrac}\right) + \valmul\powthresh\reward \atime e^{-\frac{(1-2\gdpowfrac-2\powthresh)^2\atime/2 - (1-2\gdpowfrac) l_0}{2}}\]
\end{theorem}
Critically, for sufficiently large $\atime$, we observe exponential decay in the expected attack cost beyond the per-block payout.

These bounds are far from tight. A graph of expected attack cost (excluding the per-block payouts of $\reward$) and attack success likelihood in terms of $\atime$ for an attack against $l_0=150$ is included in Figure \ref{fig:success_cost}.
\begin{figure}
	\centering
	\includegraphics[width=\textwidth]{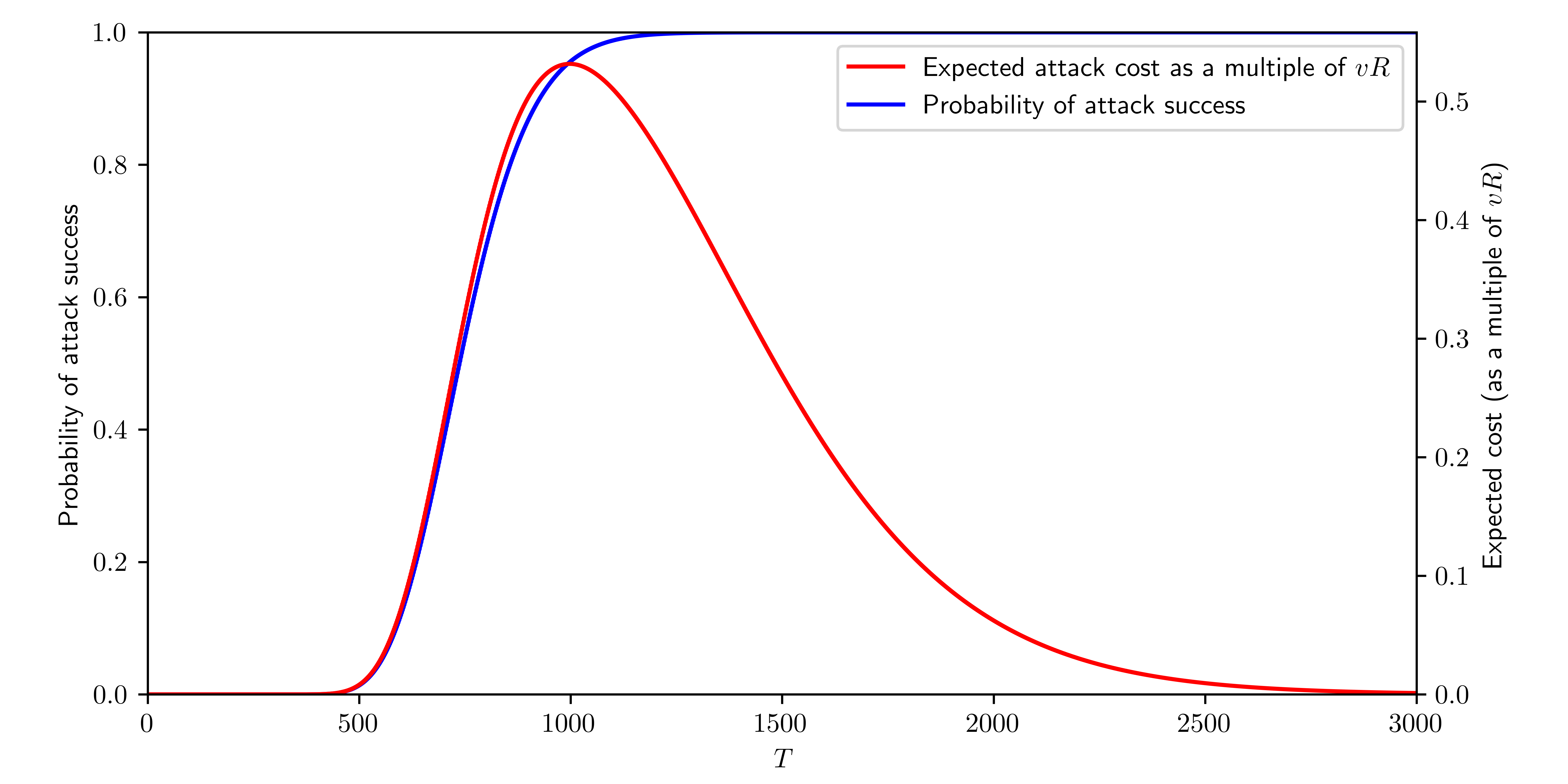}
	\caption{Probability of attack success and expected attack cost (excluding per-block payouts of $\reward$) for an attack with given $T$ and $l_0=150, \powthresh=0.05, \gdpowfrac=0.4$.}
	\label{fig:success_cost}
\end{figure}
The attack cost peaks when success is uncertain, as this is where participation/refusal has the strongest effect, and falls off rapidly as attack success becomes certain.

Two remaining points about the asymptotics are of interest. First, as $l_0$ increases, both attack success chance and attack cost converge more quickly in terms of $\frac{\atime}{l_0}$, as the expected result concentrates better. In particular, the premium over the naive cost of mining energy required to perform a range-$l_0$ attack in which $1-\gdpowfrac$ of mining power is participating is logarithmic in $l_0$, not linear. Second, due to the exponential term, $\powthresh$ does not need to be low to make the attack cost very low. For instance, $\powthresh=.2, \gdpowfrac=.25$ yields quite feasible attacks for $\frac{\atime}{l_0}$ substantially greater than $\frac{1-2\gdpowfrac}{(1-2\gdpowfrac-2\powthresh)^2}$. However, only a mild amount of successful collusion (two miners of size $.15$, for instance) is needed to prevent this attack.

\section{Relative Immunity of Proof-of-Stake: Ethereum}
Ethereum contains a significant defense against attacks of this type: slashing \cite{buterin2017casper}. As Ethereum severely penalizes stakers who validate on two incompatible chain branches, the internalized penalty for attack participation as compared to attack non-participation may be as high as the Budish payout if the attack is not expected to reduce the price of Ethereum. If attackers will lose their stake, the expected future payoff matrix is

\begin{center}
	\begin{tabular}{c||c|c|}
		& Attack succeeds & Attack fails\\
		\hline \hline
		Participate & $0$ & $0$\\
		\hline
		Refuse & $0$ & $\frac{\pf(m)}{\tpow} \valmul\reward$\\
		\hline
	\end{tabular}
\end{center}

This implies that short-range attacks (those in which attackers must possess vulnerable stake on the true chain) will cost close to the Budish estimate. Ethereum also has weak subjectivity: the property that any agent entering the network with access to a sufficiently recent honest network state can independently determine the present honest state \cite{buterin2014proof}. Assuming access to such states for entrants, this renders long-range attack impossible, so decentralization does not meaningfully reduce Ethereum attack cost.

\section{Practical Economic Considerations}
Attackers face two primary logistical difficulties: estimating the total power of participants and coordinating the attack. For PoW blockchains, both of these are ameliorated by setting up a specialized mining pool (with similar structure to that of \cite{bonneau2016buy}, albeit different intent). Such a pool could function normally until attack time, at which point it could begin sending work units for the attack chain, instead of the consensus chain. This approach has important secondary advantages: mining pools are commonly understood and easy to work with, and by offering low pool fees (or even paying pool participants slightly more than they mine, as suggested by \cite{bonneau2016buy} for a different purpose), switchover in advance of the attack could be incentivized while negligibly increasing attack cost. Stratum v2 supports header-only mining, allowing miners to attempt to mine an externally specified block, which should allow easy coordination of attack mining. The remaining coordination difficulty is in payout -- corruption payouts cannot flow through the original chain due to likelihood of devaluation. Assuming that another cryptocurrency (for instance, a PoS-based one) is expected to be unaffected by the attack, it could be used as a payment medium; otherwise, a classical medium would be needed.

Regarding pool mining, we also note that the attack proposed in Section~\ref{sec:complicated-attack} can be enhanced to provide lower payout variability with small increase in cost: in addition to the payouts given to participants who mine blocks, duplicate the defending-chain payout and split it proportionate to mining power across the pool. This guarantees participants the same consistency of payout as they would have received had they been participating in a large mining pool, while keeping cost similar, and may incentivize risk-averse/small miners to participate.

\section{Attacker Incentives}
\label{sec:attack-incentives}

The obvious purpose of a majority attack is to enable a doublespend, allowing a malicious actor to revert a transfer of a cryptocurrency pursuant to its sale, thus effectively stealing the sale price from the buyer. However, other incentives for majority attacks exist.

If blockchain technologies become more mainstream, non-currency-oriented attacks will become increasingly appealing. When substantial value may be placed on non-currency attributes of the chain (such as smart contract or DAO state, ownership of a NFT, etc.) in unpredictable ways, many non-obvious attack incentives will exist. This also makes attack attribution more difficult.

In addition to potentially being profitable, a majority attack may significantly if not completely devalue the attacked currency, and possibly others as well. The economic impact of this could be substantial. At time of writing, the global market cap of cryptocurrencies stands at about \$1 trillion USD, and various state or non-state actors may stand to gain substantially from a crash precipitated by an attack.

\subsection*{Case Study: A Bitcoin Attack}
We estimate the cost of a medium-range majority attack on Bitcoin. Distribution of control of mining capacity is a closely guarded secret, but mining power appears to be significantly spatially and internationally distributed \cite{sun2022spatial}. Assume that miners representing 60\% of mining power will participate if we set $\powthresh=1/20$ (i.e. $\gdpowfrac=0.4$). Consider an attack with a backwards range of $l_0=150$ blocks (about one day, for Bitcoin). Solving the recursion explicitly, expected attack cost for $\atime=2500$ is $<(0.01\valmul + 450)\reward$, and probability of success $\geq 1-10^{-12}$. Assuming $\valmul = 158000$ (the number of Bitcoin blocks in a little over 3 years), expected attack cost is $\leq 1942 \reward$, or a little under $12.2$k BTC/$310$M USD at time of writing. Maximum potential attack cost is substantially higher -- $2.38\valmul\reward$, or about $2.35$M BTC/$60$B USD -- but given that attack cost concentrates with exponential falloff, this may be insurable. $\powthresh$ and $\gdpowfrac$ may also be substantially overestimated here, leading to inflated cost bounds -- if we take $\powthresh=0.03, \gdpowfrac=0.2$, maximum payout is less than $820$k BTC/$21$B USD and expected payout is just over 1293 BTC/$33$M USD for a $\atime=400$ attack (with attack success chance $\geq 1-10^{-7}$).

This attack is cheap enough to be profitable, and therefore dangerous. Open Bitcoin options volume has consistently held above \$5 billion USD for three years \cite{TheBlock}, and has been substantially higher during peaks. Daily trading volume of BTC has stayed well above 200000 BTC for a similar period \cite{TheBlock}, and most BTC options are relatively short-term. Attackers able to successfully trade and liquidate options at these scales could make substantial profit.

One example of an attack strategy is as follows. First, buy $B$ ``covering'' units of BTC, and short-sell a separate $B$ units of BTC. Sell the covering units, and attack to revert that transfer. Provided that the attack succeeds, the recouped covering bitcoin may be used to cover the short position, and as its initial buy cost and sale profit are approximately equal, the net profit of the attack is equal to the gross income from the short sale. In the case of attack failure, the buy and sell costs of the covering units cancel, and cost is equal to the cost of covering the short positions minus the profit from selling them (i.e., it is as if the attackers had simply shorted BTC), less transaction costs from the purchase and sale. Even in this case, the attack may substantially devalue BTC, creating some profit from the short position.

\section{Economic Prevention and Mitigation}
\label{sec:prev}
On a protocol level, this attack is equivalent to a 51\% attack, and the same impossibility results apply to protocol-level defense. On a non-protocol level, \cite{bonneau2016buy} discusses several prevention measures. We discuss the most relevant of these, and some others, here.

\subsection{Coalitions}
In response to an incipient attack, a coalition of miners may agree on a mutual behavior contract designed to disincentivize attack participation. In particular, miners could commit to non-subgame-perfect equilibrium strategies, such as defending until a certain threshold of miners have been shown to participate in the attack, and then attacking thereafter, which stymie the given payout rule. However, credible commitment is difficult for miners, and such Nash equilibria are not robust to the attacker designing new payout schemes in response. We conjecture that given any such commitment scheme, the attacker can design a bribery system with expected payout and failure probability bounded by those of the subgame perfect Nash equilibrium and corresponding payout rule. For instance, in the above example, the attacker could increase payouts to miners until the threshold has been passed, and then remove them once it has.

Coalitions with leaders, such as some mining pools, present a separate danger to blockchains: if coalition mining power is controlled by a player whose expected future payout from attack failure is lower than $\valmul \reward$, the controller may be incentivized to make the coalition participate even if it falls above the power threshold. For instance, if the controller of some mining pool receives 1\% of the profits generated, they will be incentivized to cause the pool to participate if its power is $\leq 100\powthresh\tpow$.

\subsection{Social Consensus as Deterrence/Mitigation}
If a sufficiently broad set of participants wishes, they may fork or otherwise alter the blockchain after an attack in an attempt to reverse the effects of the attack. This is unlikely to prevent profit-taking, as it must be done rapidly enough to fix all manipulated transactions before they can be used for profit. For instance, in the attack proposed in the previous section, the sale of $B$ units of BTC must be fixed as canonical before it can be sold again in the attack chain. Moreover, this has obvious disadvantages for regular use, and does not prevent repeated sabotage attacks \cite{budish2018economic} that would effectively make the blockchain unusable (which can itself be profitable for an attacker).

\subsection{Counterattacks and the Model of Moroz et al.}
Moroz et al. \cite{moroz2020double} argue that in many cases, the threat of counterattacks renders double-spend attacks unprofitable. They analyze a model in which a single large transaction is in contention between the sender (the attacker) and the receiver (the defender), with the former wishing to establish a heaviest chain not including the transaction and the latter wishing the opposite. Each party is allowed to purchase mining power on an open market in order to attempt majority attacks to shift consensus, and take turns doing so. Moroz et al. find that under certain conditions, the only equilibrium strategy is to not attack in the first place.

However, their model (and therefore results) are inapplicable in our case, as both parties may have external rewards that do \textit{not} depend on the final status of the transaction. The success of one or more attacks will almost certainly substantially reduce the value of the attacked cryptocurrency, independent of the final state of the blockchain. This itself may be a major source of utility for the attacker (see Section \ref{sec:attack-incentives}) and disutility for the defender (the defender will recoup coins of lesser value if they are successful, and if they hold additional units of the currency, those two will be devalued by attacks).

When this assumption is changed, the no-attack equilibrium established by Moroz et al. does not hold, and attack with no response is equilibrium in many practical cases. For instance, if a single successful attack will massively devalue a currency, but the attacker can make profit from this event, the attack is incentivized \emph{even if it will be counterattacked}, and a counterattack is generally not incentivized.

\subsection{Countercorruption}
As noted in \cite{bonneau2016buy}, miners above the threshold, targets of a large double-spend, and other entities with stake in the success of blockchains may be incentivized to attempt to bribe miners to \textit{not} participate. While this may be theoretically sufficient in some cases, it is undesirable for a variety of reasons discussed in \cite{bonneau2016buy}.

\subsection{Extra Confirmations}
Historically, requiring more confirmations has been seen as a natural way to secure transactions. However, the attack cost for even long-range attacks is relatively low, and the splitting tactic noted by \cite{bonneau2016buy} likely allows evasion unless many transactions require extremely long confirmation periods.

\section{Non-Economic Prevention and Mitigation}

We see that there exist potential large-scale corruptive majority attacks that are incentive-compatible for all participants and highly profitable for the attacker, and that economic forces are insufficient to disincentivize attack initiation or participation. However, this does not preclude non-economic forces preventing these attacks in practice. We briefly discuss three potential avenues of prevention.

\subsection{Social Consensus} If a power-weighted majority of miners refuse to participate, the attack will be stymied \cite{bonneau2016buy}. However, for this to work, we require a broad coalition of miners to act against their own self-interest for philosophical or other non-economic reasons. Miners have regularly behaved selfishly in contravention of protocol \cite{yaish2022uncle}, so confidence that they will behave selflessly in this case seems misplaced. For instance, miners did not leave F2Pool in large numbers despite its timestamp manipulation, implying that miners are willing to participate in activity with negative network consequences for mildly increased personal profit.

\subsection{Force} Actors subject to the jurisdiction of a force (legal or otherwise) able to both detect attack organizers/participants and impose sufficient punishments against them are unlikely to perform/abet attacks. Most immediately, various state actors are likely to take a dim view of such an attack. In the United States, for instance, any attempt to organize such an attack would likely constitute securities market manipulation, and while the legal system is woefully unequipped to deal with the consequences of a double-spend attack, organizing (and perhaps merely participating in) one would likely incur substantial civil and criminal liability. This may be enough to dissuade most actors from performing a majority attack, but it will not be enough to disincentivize those who operate outside the bounds of the law and/or with other motivations. Even individuals subject to regulatory jurisdiction may be insufficiently deterred if the attack is conducted through privacy-preserving tools, as attribution in such a case might be difficult.

\subsection{Non-Profitability} One of the incentives to carry out such an attack is profitability. Substantially reducing the viability of profiting off of the collapse of a cryptocurrency could reduce or eliminate the profit motive for majority attacks if its price is expected to collapse by the time the heaviest chain is revised.

\section{Conclusion: What \textit{Do} We Trust?}
We clarify the analysis of miner incentives surrounding majority attacks, showing that the cost to bribe miners into attacking is far less than previously believed in proof-of-work blockchains. The cost is low enough that a corruptive majority attack could be profitable if combined with an appropriate strategy combining shorting and doublespending. Moreover, the value of the attack scales with the value of the blockchain and its associated assets/activity, and the cost of the attack decreases as the miner pool becomes more diverse and incentive-driven, implying that the danger will likely continue to increase as cryptocurrencies grow. Mitigating these attacks through mechanism design is nearly impossible, as doing so would require a method to either force small miners to internalize large social costs or prevent the attacker from profit-taking.

It is unclear whether these attacks can be prevented by non-economic factors. Further work is needed to determine the likelihood of attack in real-world contexts, but such analysis will necessarily be somewhat speculative. Even if these forces can prevent attacks in practice, the advantages of using distributed systems over more classical ledgers become substantially less clear when trust in their stability rests on widespread altruism, institutional force, or other non-economically-motivated behavior. Further analysis of the non-economic forces that may prevent majority attacks, and their implications for the usefulness and viability of blockchain technologies, is warranted.

\bibliographystyle{plain}
\bibliography{main.bib}

\begin{thebibliography}{10}

\bibitem{arnosti2022bitcoin}
Nick Arnosti and S~Matthew Weinberg.
\newblock Bitcoin: A natural oligopoly.
\newblock {\em Management Science}, 68(7):4755--4771, 2022.

\bibitem{TheBlock}
The Block.
\newblock Crypto market data, Sep 2023.

\bibitem{bonneau2016buy}
Joseph Bonneau.
\newblock Why buy when you can rent? {B}ribery attacks on bitcoin-style
  consensus.
\newblock In {\em Financial Cryptography and Data Security: FC 2016
  International Workshops, BITCOIN, VOTING, and WAHC, Christ Church, Barbados,
  February 26, 2016, Revised Selected Papers 20}, pages 19--26. Springer, 2016.

\bibitem{budish2018economic}
Eric Budish.
\newblock The economic limits of bitcoin and the blockchain.
\newblock Technical report, National Bureau of Economic Research, 2018.

\bibitem{buterin2014proof}
Vitalik Buterin.
\newblock Proof of stake: How {I} learned to love weak subjectivity, Nov 2014.

\bibitem{buterin2017casper}
Vitalik Buterin and Virgil Griffith.
\newblock Casper the friendly finality gadget.
\newblock {\em arXiv preprint arXiv:1710.09437}, 2017.

\bibitem{faliszewski2016control}
Piotr Faliszewski, J{\"o}rg Rothe, and Herv{\'e} Moulin.
\newblock Control and bribery in voting, 2016.

\bibitem{judmayer2021sok}
Aljosha Judmayer, Nicholas Stifter, Alexei Zamyatin, Itay Tsabary, Ittay Eyal,
  Peter Ga{\v{z}}i, Sarah Meiklejohn, and Edgar Weippl.
\newblock {SOK}: Algorithmic incentive manipulation attacks on permissionless
  {P}o{W} cryptocurrencies.
\newblock In {\em Financial Cryptography and Data Security. FC 2021
  International Workshops: CoDecFin, DeFi, VOTING, and WTSC, Virtual Event,
  March 5, 2021, Revised Selected Papers 25}, pages 507--532. Springer, 2021.

\bibitem{liao2017incentivizing}
Kevin Liao and Jonathan Katz.
\newblock Incentivizing blockchain forks via whale transactions.
\newblock In {\em Financial Cryptography and Data Security: FC 2017
  International Workshops, WAHC, BITCOIN, VOTING, WTSC, and TA, Sliema, Malta,
  April 7, 2017, Revised Selected Papers 21}, pages 264--279. Springer, 2017.

\bibitem{mccorry2019smart}
Patrick McCorry, Alexander Hicks, and Sarah Meiklejohn.
\newblock Smart contracts for bribing miners.
\newblock In {\em Financial Cryptography and Data Security: FC 2018
  International Workshops, BITCOIN, VOTING, and WTSC, Nieuwpoort,
  Cura{\c{c}}ao, March 2, 2018, Revised Selected Papers 22}, pages 3--18.
  Springer, 2019.

\bibitem{moroz2020double}
Daniel~J Moroz, Daniel~J Aronoff, Neha Narula, and David~C Parkes.
\newblock Double-spend counterattacks: Threat of retaliation in proof-of-work
  systems.
\newblock {\em arXiv preprint arXiv:2002.10736}, 2020.

\bibitem{nakamoto2008bitcoin}
Satoshi Nakamoto.
\newblock Bitcoin: A peer-to-peer electronic cash system.
\newblock {\em Decentralized Business Review}, page 21260, 2008.

\bibitem{shanaev2019cryptocurrency}
Savva Shanaev, Arina Shuraeva, Mikhail Vasenin, and Maksim Kuznetsov.
\newblock Cryptocurrency value and 51\% attacks: evidence from event studies.
\newblock {\em The Journal of Alternative Investments}, 22(3):65--77, 2019.

\bibitem{sun2022spatial}
Wei Sun, Haitao Jin, Fengjun Jin, Lingming Kong, Yihao Peng, and Zhengjun Dai.
\newblock Spatial analysis of global bitcoin mining.
\newblock {\em Scientific Reports}, 12(1):1--12, 2022.

\bibitem{xia2023impact}
Lirong Xia.
\newblock The impact of a coalition: Assessing the likelihood of voter
  influence in large elections.
\newblock In {\em Proceedings of the 24th ACM Conference on Economics and
  Computation}, pages 1156--1156, 2023.

\bibitem{yaish2022uncle}
Aviv Yaish, Gilad Stern, and Aviv Zohar.
\newblock Uncle maker: (time)stamping out the competition in {E}thereum.
\newblock {\em Cryptology ePrint Archive}, 2022.

\end{thebibliography}

\pagebreak
\appendix

\section{Notation Reference}
For convenience, we define the notations that appear in this paper here:
\begin{itemize}
	\item $\pool$: the set of miners.
	\item $\tpow$: the total mining power of all miners.
	\item $\pf:M\to [0, \tpow]$: the function mapping miners to their mining powers.
	\item $\atime$: the duration, in blocks, of an attack.
	\item $\reward$: the (expected) value of mining a block under the stable condition of the blockchain.
	\item $\valmul$: a multiplier such that any miner $m$ who expects per-block reward $\frac{\pf(m)}{\tpow} R$ has time-discounted total future profits $\leq \valmul \frac{\pf(m)}{\tpow} R$. That this exists (and may be reasonably bounded) follows from time-discounting, reduction in rewards over time, and hardware improvement/death.
	\item $\probflip\left(\frac{\pf(m)}{\tpow}\right)$: an upper bound on miner $m$'s estimate of the probability that their participation/non-participation in the attack will determine whether it succeeds.
	\item $\probthresh$: a threshold on $\probflip$. We design an attack that will incentivize participation by miners with $\probflip\left(\frac{\pf(m)}{\tpow}\right)\leq \probthresh$.
	\item $\powthresh$: a threshold on $\frac{\pf(m)}{\tpow}$ with similar purpose to the above.
	\item $\gdpowfrac$: the fraction of total power held by honest miners.
\end{itemize}

\section{Proofs of Results in Section \ref{sec:complicated-attack}}
\subsection*{Lemma \ref{lem:probexactbound}}
\begin{proof}
	\begin{align*}
		&\Pr\left[X^{t, l+1}_\tau\in \{1, 2\}\right]\\ =\;&\binom{\tau-t}{\frac{\tau-t+l}{2}}(1-\gdpowfrac-\powthresh)^{\frac{\tau-t+l}{2}}(\gdpowfrac+\powthresh)^{\frac{\tau-t-l}{2}}\\
		&\qquad+ \binom{\tau-t}{\frac{\tau-t+l-1}{2}}(1-\gdpowfrac-\powthresh)^{\frac{\tau-t+l-1}{2}}
		(\gdpowfrac+\powthresh)^{\frac{\tau-t-l+1}{2}}
	\end{align*}
	where the binomial coefficient $\binom{n}{k}$ is taken to be 0 if $k$ is fractional. One term will always be 0, and so letting $m=\tau - t$, we may bound this by 
	\begin{align*}
		&\sup_{l\in \{0, 1, \dots, m\}} \binom{m}{\frac{m+l}{2}}(1-\gdpowfrac-\powthresh)^{\frac{m+l}{2}}(\gdpowfrac+\powthresh)^{\frac{m-l}{2}}\\
	\end{align*}
	Note that this is effectively only over $l$ of the same parity as $\tau - t$.
	
	Let $K_l = \binom{m}{\frac{m+l}{2}} (1-\gdpowfrac-\powthresh)^{\frac{m+l}{2}} (\gdpowfrac+\powthresh)^{\frac{m-l}{2}}$. This admits a smooth extension to $l\in [0, \tau - t]$ via the gamma function. We may bound this extension (via Stirling approximation) by 
	\begin{align*}
		K_l&\leq \sup_{l\in [0, \tau - t]} \sqrt{\frac{m}{2\pi \frac{m+l}{2}\frac{m-l}{2}}} \frac{(m)^{m}}{\left(\frac{m+l}{2}\right)^{\frac{m+l}{2}}{\left(\frac{m-l}{2}\right)^{\frac{m-l}{2}}}}(1-\gdpowfrac-\powthresh)^{\frac{m+l}{2}}(\gdpowfrac+\powthresh)^{\frac{m-l}{2}}
		&\leq 
	\end{align*}
	
	We observe that for integer $l$ of parity $\tau-t$, 
	\begin{align*}
		\frac{K_{l+2}}{K_l} &= \frac{\frac{m!}{(\frac{m - (l+2)}{2})!(\frac{m + (l+2)}{2})!}}{\frac{m!}{(\frac{m-l}{2})!(\frac{m+l}{2})!}} \frac{1-\gdpowfrac - \powthresh}{\gdpowfrac+\powthresh}\\
		&= \frac{m- l}{m + l+2}\frac{1-\gdpowfrac - \powthresh}{\gdpowfrac+\powthresh}
	\end{align*}
	In particular, $K_l$ is maximized across integers at the lowest $l$ of parity $\tau-t$ s.t. the above ratio is $\leq 1$ (as this ratio is decreasing in $l$). Solving $\frac{m - l}{m + l+2}\frac{1-\gdpowfrac - \powthresh}{\gdpowfrac+\powthresh}=l$ gives $l=m(1-2(\gdpowfrac + \powthresh)) - 2(\gdpowfrac + \powthresh)$. In particular, the integer $l^{\max}$ maximizing $K_l$ is somewhere between $m(1-2(\gdpowfrac + \powthresh))-3$ and $m(1-2(\gdpowfrac + \powthresh))$. Plugging $l^*=m(1-2(\gdpowfrac + \powthresh))$ into our bound on $K_l$ gives
	\begin{align*}
		\frac{1}{2\sqrt{m}} \sqrt{\frac{1}{2(\gdpowfrac + \powthresh)(1-2(\gdpowfrac + \powthresh))}}
	\end{align*}
	Moreover, we may note that
	\begin{align*}
		\frac{K_{l^*-c}}{K_{l^*}} &= \frac{\Gamma\left(\frac{m+l^*}{2}\right)/\Gamma\left(\frac{m+l^*-c}{2}\right)}{\Gamma\left(\frac{m-l^*+c}{2}\right)/\Gamma\left(\frac{m-l^*}{2}\right)}(1-\gdpowfrac-\powthresh)^{c/2}(\gdpowfrac - \powthresh)^{-c/2}\\
		&\leq  \left(\frac{(m+l^*)(1-\gdpowfrac-\powthresh)}{(m-l^*)(\gdpowfrac + \powthresh)}\right)^{c/2}\\
		&= \left(\frac{1-\gdpowfrac - \powthresh}{\gdpowfrac + \powthresh}\right)^c
	\end{align*}
	
	Then letting $l^{\max}\in [l^*-3, l^*]$ be the maximizer of $K_l$, we obtain that \[\sup_{l} K_l = K_{l^{\max}} \leq \frac{1}{\sqrt{2\pi m}} \sqrt{\frac{1}{(\gdpowfrac + \powthresh)(1-\gdpowfrac - \powthresh)}}\left(\frac{1-\gdpowfrac - \powthresh}{\gdpowfrac + \powthresh}\right)^3\]
\end{proof}

\subsection*{Lemma \ref{lem:probminbound}}
\begin{proof}
	It suffices to observe that
	\begin{align*}
		\Pr\left[\tau=\argmin_{i\in \{t, \atime\}} X^{t, l+1}_i|X^{t, l+1}_\tau\in \{1, 2\}\right] &\leq \Pr\left[\tau=\argmin_{i\in \{\tau, \atime\}} X^{t, l+1}_i|X^{t, l+1}_\tau\in \{1, 2\}\right]\\
		&=\Pr\left[\tau=\argmin_{i\in \{\tau, \atime\}} X^{\tau, X^{t, l+1}_\tau}_i\right]\\
		&\leq \Pr\left[X^{\tau, X^{t, l+1}_\tau}_T\geq X^{\tau, X^{t, l+1}_\tau}_\tau\right]\\
		&\leq e^{-\frac{1}{2}(\atime-\tau)(1-2\gdpowfrac-2\powthresh)^2}
	\end{align*}
	by Hoeffding's inequality applied to $X^{\tau, X^{t, l+1}_\tau}_T-X^{\tau, X^{t, l+1}_\tau}_\tau$, which is by definition a sum of $T-\tau$ random variables which are independently $-1$ with probability $1-\gdpowfrac-\powthresh$ and $1$ with probability $\gdpowfrac + \powthresh$.
\end{proof}

\subsection*{Lemma \ref{lem:tech1}}
\begin{proof}
	$T=t$ is trivial. Otherwise, fix $t< k \leq  T$. Then 
	\begin{align*}
		\sum_{i=t}^T \frac{e^{-a(T-i)}}{\max(\sqrt{i-t}, 1)} & \leq \sum_{i=t}^{k-1} \frac{e^{-a(T-i)}}{\max(\sqrt{i-t}, 1)} + \sum_{i=k}^T \frac{e^{-a(T-i)}}{\max(\sqrt{i-t}, 1)}\\
		&\leq \sum_{i=t}^{k-1} \frac{e^{-a(T-(k-1))}}{\max(\sqrt{i-t}, 1)} + \sum_{i=k}^T \frac{e^{-a(T-i)}}{\sqrt{k-t}}\\
		&\leq \left(1 + \sqrt{k-1-t}\right) e^{-a(T-(k-1))} + \frac{1}{\sqrt{k-t}}\frac{1}{1-e^{-a}}
	\end{align*}
	
	Setting $k=\max \left(\lceil T-2\frac{\ln (1+\sqrt{T+1-t})}{a}+1\rceil, t+1 \right)$ yields that the above is
	\begin{align*}
		&\leq \min\left(\frac{1}{1+\sqrt{T+1-t}} + \frac{1}{\sqrt{T+1-t - 2\frac{\ln (1+\sqrt{T+1-t})}{a}}}, 1+ \frac{1}{1-e^{-a}}\right)\\
		&\leq \min\left(\frac{2}{\sqrt{T+1-t - 2\frac{\ln (1+\sqrt{T+1-t})}{a}}}, 1+ \frac{1}{1-e^{-a}}\right)
	\end{align*}
\end{proof}

\subsection*{Theorem \ref{thm:maxcost}}
\begin{proof}
	We have
	\begin{align*}
		\specpayout_{t-1, l} &=w^{\max}_{t, l+1} - w^{\max}_{t, l-1} \\
		&=\valmul\powthresh\reward\sum_{\tau=t}^\atime \Pr\left[X^{t, l+1}_\tau\in \{1, 2\}\right] \Pr\left[\tau=\argmin_{i\in \{t, \atime\}} X^{t, l+1}_i|X^{t, l+1}_\tau\in \{1, 2\}\right]\\
		&\leq \valmul\powthresh\reward\frac{\left(1-\gdpowfrac - \powthresh\right)^{5/2}}{\sqrt{2\pi}\left(\gdpowfrac + \powthresh\right)^{7/2}}
		\sum_{\tau=t}^\atime \frac{ e^{-\frac{1}{2}(\atime-\tau) (1-2\gdpowfrac-2\powthresh)^2}}{\max\left(\sqrt{\tau - t} , 1\right)}
	\end{align*}
	by Lemmas \ref{lem:probexactbound} and \ref{lem:probminbound}. Applying Lemma \ref{lem:tech1} with $a = \frac{(1-2\gdpowfrac-2\powthresh)^2}{2}$, combined with the trivial bound $\specpayout_{t-1, l}\leq w^{\max}_{t, l+1}\leq \valmul\powthresh\reward$ yields
	
	\begin{align*}
		\specpayout_{t-1, l}  &\leq \valmul\powthresh\reward\min \left( \frac{\left(1-\gdpowfrac - \powthresh\right)^{5/2}}{\sqrt{2\pi}\left(\gdpowfrac + \powthresh\right)^{7/2}} \frac{2}{\sqrt{\atime+1-t - 4\frac{\ln (1+\sqrt{\atime+1-t})}{(1-2\gdpowfrac-2\powthresh)^2}}}, 1\right)\\
		&\leq \valmul\powthresh\reward\min \left( \frac{\left(1-\gdpowfrac - \powthresh\right)^{5/2}}{\sqrt{2\pi}\left(\gdpowfrac + \powthresh\right)^{7/2}} \frac{2}{\sqrt{\atime+1-t - 4\frac{\ln (1+\sqrt{\atime})}{(1-2\gdpowfrac-2\powthresh)^2}}}, 1\right)
	\end{align*}
	
	Then the result follows directly from partitioning the indices at $T+1 - 4\frac{\ln (1+\sqrt{\atime})}{(1-2\gdpowfrac-2\powthresh)^2}$.
\end{proof}

\subsection*{Theorem \ref{thm:expectedcost}}
\begin{proof}
	Fix two walks $X$ and $Y$ on the integers starting at $l_0$. At each step, let $X$ increment with probability $\gdpowfrac$ and decrement with probability $1-\gdpowfrac$, and let $Y$ increment with probability $\gdpowfrac+\powthresh$ and decrement with probability $1-\gdpowfrac-\powthresh$. $X$ will correspond to game state (until hitting 0 -- i.e. the distributions of $l_t$ and $X_t$ are identical over positive integers), and $Y_t$ is a virtual walk corresponding to the $w^{\max}$ recurrence. 
	
	Fix $k>l_0$. By Hoeffding's inequality on $X_t$, the probability that $l_t\geq k$ for given $t$ is $\leq e^{-\frac{(k+(1-2\gdpowfrac)t)^2}{2t}}\leq e^{-(1-2\gdpowfrac)(\frac{(1-2\gdpowfrac)t}{2}+k)}$. Taking a union bound across all $t$, the probability that there exists a timestep $t$ with $l_t\geq k$ is
	\begin{align*}
		&\leq \sum_{i=1}^\atime e^{-(1-2\gdpowfrac)(\frac{(1-2\gdpowfrac)t}{2}+k)}\\
		&< \sum_{t=1}^{\infty} e^{-(1-2\gdpowfrac)(\frac{(1-2\gdpowfrac)t}{2}+k)}\\
		&<\frac{e^{-(1-2\gdpowfrac)k}}{1-e^{-\frac{(1-2\gdpowfrac)^2}{2}}}
	\end{align*}
	
	We may observe by its recursion that 
	\begin{align*}
		w^{\max}_{t, l} &= \valmul\powthresh\reward \Pr[\min_{i\in [\atime]} Y_i > 0|Y_t\leq l]\\
		&\leq \valmul\powthresh\reward\Pr[ Y_{\atime} > 0|Y_t\leq l]\\
		&\leq
		\begin{cases}
			(1-2\gdpowfrac - \powthresh)(\atime - t) < l: & \valmul\powthresh\reward \\
			(1-2\gdpowfrac - \powthresh)(\atime - t) \geq l: & \valmul\powthresh\reward  e^{-\frac{(l - (1-2\gdpowfrac - \powthresh)(\atime - t))^2}{2(\atime - t)}}\\
		\end{cases}\\
		&\leq \valmul\powthresh\reward e^{l(1-2\gdpowfrac - 2\powthresh) - \frac{(1-2\gdpowfrac - 2\powthresh)^2(\atime - t)}{2}}
	\end{align*}
	where the concentration follows from Hoeffding's inequality. We may also bound the probability that the attack has not succeeded by the start of step $t$, conditioned on there not existing a timestep $t$ with $l_t\geq k$, as 
	\begin{align*}
		\Pr[X_{t-1}>0 | \sup_{i\in [\atime]} X_i < k] &\leq \Pr[X_{t}\geq 0 | \sup_{i\in [\atime]} X_i < k] \\
		&\leq \Pr[X_{t}\geq 0]\\
		&\leq e^{l_0(1-2\gdpowfrac) - \frac{(1-2\gdpowfrac)^2 t}{2}}
	\end{align*}
	by the same.
	
	Finally, we observe that \[\specpayout_{t-1, X_{t-1}}\leq w^{\max}_{t, X_{t-1}+1}\leq  e^{(X_{t-1}+1)(1-2\gdpowfrac - 2\powthresh) - \frac{(1-2\gdpowfrac - 2\powthresh)^2(\atime - t)}{2}}\] In particular, conditioned on $\forall i\in [\atime]:X_i<k$, we have 
	\[\specpayout_{t-1, l}\Pr[X_{t}\geq 0 | \forall i\in [\atime]:X_i<k]\leq e^{l_0 (1-2\gdpowfrac) +  k(1-2\gdpowfrac-2\powthresh) - \atime (1-2\gdpowfrac-2\powthresh)^2/2}\]
	
	Then we bound expected attack cost $C$ (excluding the per-attacking-block payout) as 
	\begin{align*}
		\E[C] &\leq \E[C] \Pr[\exists t\in [\atime]: X_t\geq k] + \E[C|\forall t\in [\atime]: X_t < k] \\
		&\leq \E[C] \Pr[\exists t\in [\atime]: X_t\geq k] + \sum_{t=1}^\atime \specpayout_{t-1, X_{t-1}<k} \Pr[X_{t-1}>0 | \sup_{i\in [\atime]} X_i < k]\\
		&\leq \valmul\powthresh\reward\left[\atime \frac{e^{-(1-2\gdpowfrac)k}}{1-e^{-\frac{(1-2\gdpowfrac)^2}{2}}}+\atime e^{l_0 (1-2\gdpowfrac) +  k(1-2\gdpowfrac-2\powthresh) - \atime (1-2\gdpowfrac-2\powthresh)^2/2}\right]
	\end{align*}
	Solving $-(1-2\gdpowfrac)k = l_0 (1-2\gdpowfrac) + k(1-2\gdpowfrac-2\powthresh) - \atime(1-2\gdpowfrac-2\powthresh)^2/2$ yields $k=\frac{(1-2\gdpowfrac-2\powthresh)^2\atime/2 - (1-2\gdpowfrac) l_0}{2-4\gdpowfrac-2\powthresh}$ and therefore total expected cost 
	\begin{align*}
		&\leq \valmul\powthresh\reward \atime e^{-(1-2\gdpowfrac) \frac{(1-2\gdpowfrac-2\powthresh)^2\atime/2 - (1-2\gdpowfrac) l_0}{2-4\gdpowfrac-2\powthresh}}\\
		&\leq \valmul\powthresh\reward \atime e^{-\frac{(1-2\gdpowfrac-2\powthresh)^2\atime/2 - (1-2\gdpowfrac) l_0}{2}}
	\end{align*}
	
	The expected number of attacking blocks mined is bounded by $\frac{l_0}{2} + \frac{1}{2}\frac{l_0}{1-2\gdpowfrac}$ (as the biased random walk is expected to last $\frac{l_0}{1-2\gdpowfrac}$ steps, and if the attack lasts $k$ blocks, the attackers mine at most $\frac{l+k}{2}$ of them. Then total expected attack cost is bounded by 
	\[\reward\left(\frac{l_0}{2} + \frac{1}{2}\frac{l_0}{1-2\gdpowfrac}\right) + \valmul\powthresh\reward \atime e^{\frac{(1-2\gdpowfrac-2\powthresh)^2\atime/2 - (1-2\gdpowfrac) l_0}{2}}\]
\end{proof}

\end{document}